\documentclass[12pt]{amsart}

\usepackage[mathscr]{eucal}
\usepackage{amsthm,amsmath,amssymb,amscd}
\usepackage{epsfig}
\newtheorem{theo}{Theorem}
\newtheorem{definition}{Definition}
\newtheorem{lemma}{Lemma}
\newtheorem{cor}{Corollary}
\newtheorem{prop}{Proposition}
\theoremstyle{definition}
\newtheorem{remark}{Remark}
\theoremstyle{plane}

\def \beq{ \begin{equation} }
\def \eeq{\end{equation}}

\def \ind{\mbox{ind}}
\begin{document}
\title{Central Configurations and Total Collisions for Quasihomogeneous n-Body Problems}

\renewcommand\baselinestretch{}
\authors{Florin Diacu$^{*}$, Ernesto P\'erez-Chavela$^{**}$ and Manuele
Santoprete$^{***}$}

\address{\noindent $^{*}$ Department of Mathematics and Statistics,
University of Victoria, Victoria Canada}
\address{\noindent $^{**}$ Departamento de Matem\'aticas, UAM--Iztapalapa, A.P. 55--534,
09340 Iztapalapa, Mexico, D.F., Mexico}
\address{\noindent $^{***}$ Department of Mathematics,
University of California, Irvine, USA}

\email{diacu@math.uvic.ca, epc@xanum.uam.mx, msantopr@math.uci.edu}






\begin{abstract}
We consider $n$-body problems given by potentials
of the form ${\alpha\over r^a}+{\beta\over r^b}$ with 
$a,b,\alpha,\beta$ constants, $0\le a<b$. To analyze the dynamics
of the problem, we first prove some properties related
to central configurations, including a generalization of
Moulton's theorem. Then we obtain several qualitative 
properties for collision and near-collision orbits in the Manev-type case
$a=1$. At the end we point out some new relationships between central configurations, relative equilibria, and homothetic 
solutions. 
\date{\large{\today}}
\end{abstract}

\maketitle

\markboth{F. Diacu, E. P\'erez-Chavela and M. Santoprete}{The Quasihomogeneous $n$-Body Problem}


\section{Introduction}

The $n$-body problem studied here is given by a potential of the form ${\alpha\over r^a}+{\beta\over r^b}$, where $r$ is the distance between bodies and $a,b,\alpha,\beta$ are constants, $0\le a<b$ (see \cite{diacu96, ernesto98}). In the first part of the paper we treat the general problem, and in the second part we focus on the case $a=1$. The function ${\alpha\over r^a}+{\beta\over r^b}$, called quasihomogeneous because of being the sum of homogenous functions of different degrees, generalizes classical potentials, such as those of Newton, Coulomb, Birkhoff, Manev, Van der Waals, Libhoff, Schwarzschild, and Lennard-Jones. Thus, the applicability of the quasihomogeneous $n$-body problem ranges from celestial mechanics and atomic physics to chemistry and crystallography.

Although many properties of the Newtonian $n$-body problem have a correspondent in the homogeneous case, this is not true for nonhomogeneous potentials. On one hand, the transposition of known results is far from trivial; on the other, new properties show up.

An intriguing aspect we will point out in this paper refers to central configurations, which are crucial for understanding the dynamics of the $n$-body problem (see \cite{saari80}). The central configurations of the quasihomogeneous potential are in a certain relationship with the central configurations of the homogeneous functions that form this potential. Thus, we will introduce here the notion of simultaneous central configuration and will investigate its connection with the classical concept.

In Section 2, we define the quasihomogeneous $n$-body problem and write down the equations of motion. In Section 3, we introduce the concepts of central and simultaneous central configuration, the latter being specific to quasihomogeneous potentials. Section 4 deals with collinear central configurations. Using critical point theory, we prove a generalization of Moulton's theorem by showing that the number of collinear central configurations of $n$ bodies is $n!/2$. Starting with Section 5, we restrict our study to Manev-type problems, \cite{diacu00}, i.e.\ those given by potentials of the form ${\alpha\over r}+{\beta\over r^b}$, and show that there are exactly two planar central configurations in the 3-body case. Section 6 introduces a framework for the study of collision and near-collision orbits, which is performed in Sections 7 and 8. We study in detail the network of collision solutions and determine the relationship between central configurations, on one hand, and relative equilibria and homothetic orbits, on the other hand. 
It is important to note that if in the homogeneous case the correspondence between central configurations and homothetic solutions is one-to-one, this fails to be the case in the quasihomogeneous problem. The relationship between central configurations and relative equilibria remains unchanged, i.e.\ one-to-one, in the quasihomogeneous case. For Manev-type potentials, homothetic orbits are less likely than in the Newtonian case, in the sense that they show up only for simultaneous central configurations.

\section{The Quasihomogeneous $n$-Body Problem}

We will start with defining the planar quasihomogeneous $n$-body problem.  Consider the linear space 
\beq \Omega=\{{\bf r}=({\bf
r}_1,\ldots,{\bf r}_n)\in({\mathbb R}^2)^n|\sum_{i=1}^n m_i{\bf
r}_i=0\}, \eeq
where $m_i>0, i=1,2,\dots,n$, are the masses of the $n$ bodies and
${\bf r}_i, i=1,2,\dots,n$, represent their coordinates. Notice that $\sum_{i=1}^n m_i{\bf r}_i=0$ fixes the centre of mass
at the origin of the coordinate system.
Let \beq \Delta_{ij}=\{({\bf r}_1,\ldots,{\bf
r}_n)\in \Omega|{\bf r}_i={\bf r}_{j}\}; \quad
\Delta=\bigcup_{i,j} \Delta_{ij}. \eeq We call $\Delta$ the collision set. The potential
$\,U\,$ of the system is a function defined on the configuration
space $\tilde \Omega=\Omega\setminus \Delta$ and  is given by $$ U
= W + V,$$ where $W$ is a homogeneous function of degree $-a$,
$a\ge 0$,
\begin{equation}\label{homa}
W({\bf r}_1,\ldots,{\bf r}_n)\,=\,\sum_{i<j}\frac{m_im_j}{\|{\bf r}_i-{\bf r}_j\|^a},
\end{equation}
and $V$ is a homogeneous function of degree
$-b$, $b>a$,
\begin{equation}\label{homb}
V({\bf r}_1,\ldots,{\bf r}_n)\,=\,\sum_{i<j}\frac{m_im_j}{\|{\bf r}_i-{\bf r}_j\|^b}.
\end{equation}
The equations of motion of the $n$ bodies define a vector field $X$  on the tangent bundle $T(\tilde\Omega)$.
The configuration space of the system is $\tilde\Omega$ and the cotangent bundle is $T^*(\tilde\Omega)$.
Let $\;{\bf p}\,=\,M^{-1}\dot{\bf r}\;$
be the linear momentum of the system of particles, where $\,M\,$
is the diagonal matrix
$\,M\,=\,\mbox{diag}\,(m_1,m_1,m_2,m_2,\dots,m_n,m_n).\;$ Then the
equations of motion  can be written as a Hamiltonian system,
\begin{equation}\label{hamsys}
\begin{split}
   \dot{\bf r} &=\;\displaystyle\frac{\partial H}{\partial{\bf p}}\\
   \dot{\bf p} &= -\,\displaystyle\frac{\partial H}{\partial{\bf r}},
\end{split}
\end{equation}
where $\,H:T^*(\tilde\Omega)\to{\rm I\!R}\,$ is the
Ham\-il\-to\-ni\-an function given by
\begin{equation}\label{hamf}
   H({\bf r}, {\bf p})\,=\,\frac{1}{2}\,{\bf p}^t\,M^{-1}{\bf p}\,-\,U({\bf r}).
\end{equation}
\noindent
Here $T=\frac{1}{2}\,{\bf p}^t\,M^{-1}{\bf p}\;$ is the kinetic
energy. The total energy $\,H\,$ is a first integral for the
system (\ref{hamsys}); this means that $T-U=h$ (constant)  along any orbit. Other integrals are given by the linear momentum,
$\sum_{i=1}^nm_i{\dot{\bf r}}_i$, and by the angular momentum,
$J:T\rightarrow {\bf R}$, defined as \beq J({\bf r},{\bf
v})=\sum_{i=1}^n m_i{\bf r}_i\times {\bf v}_i. \eeq
Notice that the relationships for the centre of mass, 
$\sum_{i=1}^nm_i{\bf r_i}=0$, and linear momentum,
$\sum_{i=1}^nm_i{\dot{\bf r}}_i=0$, together with the
energy integral, $T-U=h$, reduce the dimension of the Hamiltonian
system (\ref{hamsys}) from $4n$ to $4n-5$.
We also introduce the scalar product, \beq \langle {\bf
r}, {\bf \tilde r } \rangle={\bf r}^t M {\bf \tilde r}, \eeq
which allows us to write the moment of inertia as \beq
I=\langle {\bf r}, {\bf r }\rangle=\sum_{i=1}^nm_i\|{\bf r}_i\|^2.
\eeq

\section{Central Configurations}

Central configurations play a crucial role for understanding the dynamics of $n$-body problems \cite{saari80}. In particular,
they have led to important theoretical investigations, such as
Saari's conjecture, which has remained open for more than
three and a half decades \cite{diacu05}, and are connected to
Smale's 6th problem \cite{smale98}, originally proposed by 
Wintner in 1941, \cite{wintner41} (see also \cite{meyer92, ernesto96}). In this section we will
define central configurations and analyze the particular aspects
this concept encounters in the quasihomogeneous case.

\begin{definition}
A configuration ${\bf r}\in \tilde\Omega$ is called central if there is a constant $\sigma$ such that
\beq
\nabla U({\bf r})=\sigma\nabla I({\bf r}).
\eeq
\end{definition}
Using the fact that the functions $W$ and $V$ are homogeneous of degree $-a$ and $-b$, respectively, and applying Euler's theorem for homogeneous functions, we find that
\beq
\sigma=\frac{-aW({\bf r})-bV({\bf r})}{2 I({\bf r})}.
\eeq
\begin{definition}
We call ${\bf r}\in \tilde\Omega$ a simultaneous central configuration for the potentials $W$ and $V$ if there are constants $\sigma_1$ and $\sigma_2$ such that
\[
\nabla W({\bf r})=\sigma_1\nabla I({\bf r}) \quad \mbox{and} \quad \nabla V({\bf r})=\sigma_2\nabla I({\bf r}).
\]
\end{definition}

Using the fact that $W$ and $V$ are homogeneous functions of degree $-a$ and $-b$, respectively, we find that
\beq
\sigma_1=\frac{-aW({\bf r})}{2I({\bf r})} \quad \mbox{and}\quad \sigma_2=\frac{-bV({\bf r})}{2I({\bf r})}.
\eeq
Note that if ${\bf r}$ is a simultaneous central configuration for $W$ and $V$, then ${\bf r}$ is also a central configuration for $U=V+W$.
The converse is not necessarily true.

Let
\[
S_{I_0}=\{{\bf r}\in \Omega|\langle {\bf r},{\bf r} \rangle=I_0\}
\]
be the sphere relative to the metric given by the
scalar product, and denote by 
\[S^*_{I_0}=S_{I_0}\setminus \Delta
=\{{\bf r}\in\tilde \Omega|\langle{\bf r},{\bf r}\rangle=I_0\}\]
this sphere minus the collision set.
Then the central configurations with moment of inertia $I_0$ can 
also be defined as the critical points of $U_{S_I}$, where
$U_{S_I}:S^*_{I}\rightarrow {\mathbb R}$ is the restriction of the
potential $U$ to $S^*_{I_0}$. Denote by $C_n$ the set of
central configurations of the quasihomogeneous $n$-body problem. 
\begin{definition}
We say that two relative equilibria in $S^*_{I_0}$ are equivalent (and
belong to the same equivalence class) if they can be made
congruent by the induced $S^1$ action on $S^*_{I_0}$, that is, if
one is obtained from the other by a rotation. 
\end{definition}
Let $\tilde C_n$ denote the set of equivalence classes of central configurations. Note that this definition differs from the one used in 
the Newtonian case (see \cite{abraham,smale71}), where two central configuration are called equivalent when one can be obtained from the other  by a rotation and/or a homothety. This change is necessary in the quasihomogeneous case because the set $C_n$ is invariant under the action of the group $S^1$, but not necessarily under the action of homotheties (see Section 7).

Clearly, $I$ and $\Delta$ are invariant under the action of $S^1$.
Thus, we can conclude that $S^*_{I_0}$ is diffeomorphic to the $(2n-3)$-dimensional sphere $S^{2n-3}$ (which is actually an ellipsoid $E^{2n-3}$) with all the points $\Delta $ removed, that is,
\[S^*_{I_0}=E^{2n-3}\setminus (E^{2n-3}\cap\Delta)\approx S^{2n-3}\setminus (S^{2n-3}\cap\Delta).\]
 Since $U_{S_I}$ is invariant under the action of $S^1$, it defines a map $\tilde U_{S_I}:S^*_{I_0}/S^1\rightarrow {\mathbb R}$.
 If we let $\pi_n:S^*_{I_0}\rightarrow S^*_{I_0}/S^1$ denote the canonical projection, $\tilde\Delta=\pi(E^{2n-3}\cap\Delta)$, and recalling that $E^{2n-3}/S^1\approx S^{2n-3}/S^1\approx {\mathbb C}P^{n-2}$ (the complex projective space), we are led to investigate the critical points of $\tilde U_{S_I}:{\mathbb C}P^{n-2}\setminus \tilde\Delta \rightarrow {\mathbb R}$.
 
Consequently we can show that the set of equivalence classes of central configurations with fixed moment of inertia $I_0$ is given by the set of critical points of the map $\tilde U_{S_I}:{\mathbb C}P^{n-2}\setminus \tilde\Delta\rightarrow {\mathbb R}$. More precisely, we have proved the following property:

\begin{prop}
For any choice of masses in the planar $n$-body problem with a quasihomogeneous potential, $n\geq 2$, the set of equivalence classes of central configurations with moment of inertia $I_0$ is diffeomorphic with the set of critical points of the map $\tilde U_{S_I}:{\mathbb C}P^{n-2}\setminus\tilde\Delta\rightarrow {\mathbb R}$.
\end{prop}

\section{Moulton's Theorem for Quasihomogeneous Potentials}

We will now study collinear central configurations and, using critical point theory, will calculate the number of classes of such configurations for any number $n$ of bodies. The goal of this section is to prove the following result, which generalizes a theorem obtained by Forest Ray Moulton in 1910, \cite{moulton10}.

\begin{theo}
For any choice of masses in the $n$-body problem with a quasihomogeneous potential, $U$, and any given moment of inertia, $I_0$, there are exactly $n!/2$ classes of collinear central configurations. In other words, there are $n!/2$ classes of central configurations ${\bf r}=({\bf r}_1,\ldots,{\bf r}_n)$, where all ${\bf r}_i$ belong to the same straight line through the origin.
\end{theo}
In preparation for the proof, choose some line $l$  in ${\mathbb R}^2$. This defines a subset $\Omega_l\subset\Omega$ of ${\bf r}=({\bf r}_1,\ldots,{\bf r}_n)$ such that each ${\bf r}_i$  is on the line $l$. Let $S_l=S_{I_0}\cap \Omega_l$ and $S_l^*=S_l\setminus\cap(S_l\cap \Delta)$. When $S^1$ acts on $S_{I_0}$, only the rotation by $\pi$ radians leaves $S_l$ invariant. Thus the group ${\mathbb Z}_2$ acts on $S_l$, and on  the quotient we have
${\mathbb R}P^{n-2}\setminus\tilde\Delta\subset {\mathbb C}P^{n-2}\setminus\tilde\Delta\xrightarrow{\tilde U_{S_I}} {\mathbb R}$, where ${\mathbb R}P^{n-2}=S_l\backslash {\mathbb Z}_2$ is the real projective space, naturally contained in ${\mathbb C}P^{n-2}$. Here $\tilde U_{S_I}$ is induced by the potential energy.
From these considerations we obtain:
\begin{lemma}
The set of equivalence classes of collinear central configurations with moment of inertia $I_0$ is diffeomorphic to the set of critical points of $\tilde U_{S_I}:{\mathbb C}P^{n-2}\setminus \tilde\Delta\rightarrow {\mathbb R}$ that lie in ${\mathbb R}P^{n-2}\setminus \tilde\Delta \subset{\mathbb C}P^{n-2}\setminus\tilde\Delta$.
\end{lemma}
So in order to describe the collinear central configurations, it is sufficient to obtain the critical points of the potential that lie in the real projective space. 
In general, a critical point of a function restricted to a submanifold is not necessarily a critical point of the function on the ambient manifold. However, we have the following result:
\begin{prop}\label{restriction}
If ${\bf r}\in{\mathbb R}P^{n-2}\setminus\tilde\Delta$ is a critical point of $\tilde U_{S_I}:{\mathbb R}P^{n-2}\setminus\tilde\Delta\rightarrow {\mathbb R}$, then ${\bf r}$ is also a critical point of $\tilde U_{S_I}:{\mathbb C}P^{n-2}\setminus\tilde\Delta\rightarrow {\mathbb R}$.
\end{prop}
To prove this, we first need to know the derivatives of the potential function, which are given below.
\begin{lemma}\label{derivatives}
For given masses $m_1,\ldots, m_n$ and $U=W+V$,
\begin{enumerate}
\item The first derivative of $U:\tilde\Omega\rightarrow {\mathbb R}$ is 
\[\begin{split}
DU({\bf r})({\bf v})=&-a\sum_{i\neq j}\frac{m_im_j}{\|{\bf r}_i-{\bf r}_j\|^{a+2}}({\bf r}_i-{\bf r}_j,{\bf v}_i-{\bf v}_j)\\
&-b\sum_{i\neq j}\frac{m_im_j}{\|{\bf r}_i-{\bf r}_j\|^{b+2}}({\bf r}_i-{\bf r}_j,{\bf v}_i-{\bf v}_j)\end{split}\]
for ${\bf v}\in \Omega$.
\item The $2\mbox{nd}$ derivative is
\[\begin{split}
&D^2U({\bf r})({\bf v},{\bf w})= a\sum_{i\neq j}\frac{m_im_j}{\|{\bf r}_i-{\bf r}_j\|^{a+2}}\\
&\cdot\left( \frac{a+2}{\|{\bf r}_i-{\bf r}_j\|^{2}}({\bf r}_i-{\bf r}_j,{\bf v}_i-{\bf v}_j)({\bf r}_i-{\bf r}_j,{\bf w}_i-{\bf w}_j)-({\bf v}_i-{\bf v}_j,{\bf w}_i-{\bf w}_j)\right)\\
&+ b\sum_{i\neq j}\frac{m_im_j}{\|{\bf r}_i-{\bf r}_j\|^{b+2}}\\
&\cdot\left( \frac{b+2}{\|{\bf r}_i-{\bf r}_j\|^{2}}({\bf r}_i-{\bf r}_j,{\bf v}_i-{\bf v}_j)({\bf r}_i-{\bf r}_j,{\bf w}_i-{\bf w}_j)-({\bf v}_i-{\bf v}_j,{\bf w}_i-{\bf w}_j)\right),
\end{split}
\]where ${\bf v},{\bf w}\in \Omega$.
\item The $2\mbox{nd}$ derivative of the restriction $U:S^*_{I_0}\rightarrow {\mathbb R}$ is:
\[
D^2U/(S^*_{I_0})({\bf r})({\bf v},{\bf w})=D^2U({\bf r})({\bf v},{\bf w})+\frac{aW({\bf r})+bV({\bf r})}{I_0}\langle {\bf v},{\bf w}\rangle.
\]
\end{enumerate}
Here $(\cdot,\cdot)$ denotes the usual inner product in ${\mathbb
R}^2$, $\|\cdot\|$ the norm in ${\mathbb R}^2$, and $I$ the moment of
inertia. The same formulas are valid in ${\mathbb R}$, ${\mathbb
R^2}$ and ${\mathbb R}^3$.
\end{lemma}
\begin{proof}
All the equations above can be derived by differentiating in local Cartesian coordinates.
\end{proof}

Now we can give a proof of Proposition \ref{restriction}.
For ${\bf v}_i\in{\mathbb R}^2$, let ${\bf v}_i=(v_i^\parallel,v_i^\perp)$ where $v_i^\parallel\in l$ and
$v_i^\perp\in l^\perp$. Then we can write ${\bf v}=({\bf v}^\parallel,{\bf v}^\perp)$ with ${\bf v}^\parallel=(v_1^\parallel,\ldots,v_n^\parallel)$, ${\bf v}^\perp=(v_1^\perp,\ldots,v_n^\perp)$ for each ${\bf v}\in \Omega$. If ${\bf r}\in S_l\subset S_{I_0}$, ${\bf r}\notin\Delta$, we have $T_{\bf r}(S_{I_0})=\{{\bf v}\in \Omega|\langle {\bf v},{\bf r}\rangle=0\}$ and $T_{\bf r}(S_{l})=\{{\bf w}\in \Omega_l|\langle {\bf w}, {\bf r}\rangle=0\}$ where, as usual, $\Omega$ is endowed with the mass scalar product. If ${\bf v}\in T_{\bf r}(S_{I_0})$ and ${\bf v}=({\bf v}^\parallel,{\bf v}^\perp)$, then ${\bf v}^\parallel\in \Omega_l$ and $\langle {\bf v},{\bf r}\rangle=\langle {\bf v}^\parallel,{\bf r}\rangle$. Thus ${\bf v}^\parallel\in T_{\bf r}(S_l)$, because $\langle {\bf v},{\bf r} \rangle=0$ implies $\langle {\bf v}^\parallel, {\bf r}\rangle=0$.

By Lemma \ref{derivatives} it follows that if ${\bf r}\in S_l\setminus\Delta$ and ${\bf v }\in T_{\bf r}(S_{I_0})$, then $DU({\bf r})({\bf v})=DU({\bf r})({\bf v^\parallel})$. So $DU({\bf r})({\bf v^\parallel})=0$ implies that $DU({\bf r})({\bf v})=0$.
This completes the proof of Proposition \ref{restriction}.

\begin{lemma}
${\mathbb R}P^{n-2}$ has $n!/2$ components.
\end{lemma}
\begin{proof}
Let ${\bf r}=({\bf r}_1.\ldots,{\bf r}_n)\in S_l\setminus\Delta$ and let ${\bf r}_1<\ldots<{\bf r}_n\in {\mathbb R}$ (we use the fact that the ${\bf r}_i$ are all distinct). Let $\alpha=(i_1,\ldots,i_n)$ be an arbitrary permutation of the numbers $(1,2,\ldots,n)$. If we apply the permutation to the initial vector ${\bf r}$, we map it to a different component defined uniquely by the given permutation. Therefore the set $S_l\setminus\Delta$ has $n!$ components and the quotient space ${\mathbb R}P^{n-2}\setminus\Delta$ has $n!/2$ components.
\end{proof}

We can now prove Moulton's theorem for quasihomogeneous potentials. By applying part (2) and (3) of Lemma \ref{derivatives}, we see that $D^2U/(S_l\setminus\Delta)$ is a positive definite form, and consequently  $\tilde U$ is convex. This shows that $\tilde U$ has a unique minimum in each component of ${\mathbb R}P^{n-2}$. Thus there are $n!/2$ critical points and hence $n!/2$ central configurations.

\begin{remark}
We have identified the symmetric central configurations, otherwise the number of classes of central configurations would be $n!$.
\end{remark}

\section{Planar Central Configurations}

In this and subsequent sections, we will restrict our study to Manev-type quasihomogeneous potentials, namely those $U$ for which $a=1$ (see also \cite{diacu00}). They form an important class of quasihomogeneous potentials, derived from the Manev law, which can explain the perihelion advance of the planet Mercury within the framework of classical mechanics (for more details see \cite{diacu96} and \cite{craig99}). Since for any planar central configuration in the Manev-type three body problem, the mutual distances are geometrically independent, we can solve the equations defining the central configurations in terms of the mutual distances. To be precise, we state here a result whose proof can be found in \cite{CLP03}.

\begin{lemma}
\label{lemaa} Let $u=f(\mathbf{x})$ be a function with
$\mathbf{x}=(x_1,x_2,\dots,x_n)$,  $x_1=g_1(\mathbf{y})$,
$x_2=g_2(\mathbf{y})$,$\dots$, $x_n=g_n(\mathbf{y})$,
$\mathbf{y}=(y_1,y_2,\dots,y_m)$ and $m\geqslant n$.

If $\mbox{rank}\:(A)=n$, where {\small
\begin{equation}
A = \left(\begin{array}{ccc} \frac{\partial x_1}{\partial y_1} &
\dots
& \frac{\partial x_n}{\partial y_1}\\
\vdots & \ddots & \vdots \\
\frac{\partial x_1}{\partial y_m} & \dots & \frac{\partial
x_n}{\partial y_m}\end{array}\right) \ ,
\end{equation}} then
$\nabla f(\mathbf{x})=\mathbf{0}$ if and only if $\nabla
u(\mathbf{y})=\mathbf{0}$.
\end{lemma}

Let us now consider the 3-body case, and for this purpose we will use the notation $\mathbf{r}_i=(q_{i1},q_{i2})$ for $i=1,2,3$. From 
Lemma~\ref{lemaa} we have that if $\mbox{rank}\: (A)=3$, where
{\small
$$
A=\left(\begin{array}{ccc} \frac{\partial r_{12}}{\partial q_{11}}
& \frac{\partial r_{13}}{\partial
q_{11}}&\frac{\partial r_{23}}{\partial q_{11}}\\[6pt]
\frac{\partial r_{12}}{\partial q_{12}} & \frac{\partial
r_{13}}{\partial
q_{12}}&\frac{\partial r_{23}}{\partial q_{12}}\\[6pt]
\frac{\partial r_{12}}{\partial q_{21}} & \frac{\partial
r_{13}}{\partial
q_{21}}&\frac{\partial r_{23}}{\partial q_{21}}\\[6pt]
\frac{\partial r_{12}}{\partial q_{22}} & \frac{\partial
r_{13}}{\partial
q_{22}}&\frac{\partial r_{23}}{\partial q_{22}}\\[6pt]
\frac{\partial r_{12}}{\partial q_{31}} & \frac{\partial
r_{13}}{\partial
q_{31}}&\frac{\partial r_{23}}{\partial q_{31}}\\[6pt] \frac{\partial
r_{12}}{\partial q_{32}} & \frac{\partial r_{13}}{\partial
q_{32}}&\frac{\partial r_{23}}{\partial
q_{32}}\end{array}\right)=\left(\begin{array}{ccc}
\frac{q_{11}-q_{21}}{r_{12}} & \frac{q_{11}-q_{31}}{r_{13}}&0\\[6pt]
\frac{q_{12}-q_{22}}{r_{12}} & \frac{q_{12}-q_{32}}{r_{13}}&0\\[6pt]
-\frac{q_{11}-q_{21}}{r_{12}} & 0&\frac{q_{21}-q_{31}}{r_{23}}\\[6pt]
-\frac{q_{12}-q_{22}}{r_{12}} & 0&\frac{q_{22}-q_{32}}{r_{23}}\\[6pt]
0 &
-\frac{q_{11}-q_{31}}{r_{13}}&-\frac{q_{21}-q_{31}}{r_{23}}\\[10pt] 0
&
-\frac{q_{12}-q_{32}}{r_{13}}&-\frac{q_{22}-q_{32}}{r_{23}}\end{array}\right)\
,
$$}
then {\small
$\nabla U(\mathbf{r}_1, \mathbf{r}_2, \mathbf{r}_3)=0\ 
\mbox{if and only if}\ \nabla U(r_{12},r_{13},r_{23})=0.\ $}

Some straightforward computations show that the $\mbox{rank}\:(A)=3$ if
and only if {\small
$$
\det\left(\begin{array}{ccc} q_{11} &q_{12} &1\\  q_{21} &q_{22}
&1\\ q_{31} &q_{32} &1
\end{array}\right)\neq 0\ .
$$}
This determinant is twice the oriented area of the triangle formed
by the 3 particles. In short, if $\mathbf{r}_1$, $\mathbf{r}_2$
and $\mathbf{r}_3$ are not collinear, then $\nabla
U(r_{12},r_{13},r_{23})=\mathbf{0}$ if and only if $\nabla
U(\mathbf{r}_1, \mathbf{r}_2, \mathbf{r}_3)=0$.

Using Lemma~\ref{lemaa} in order to
find the planar central configurations, we first need to solve the
equation
\begin{equation}\label{ccdis}
\nabla U=\sigma\nabla I
\end{equation}
in terms of the mutual distances $r_{ij}$, taking into 
account the fact that the moment of inertia, $I$, can be 
written in terms of the mutual distances as 
$I=(1/{\tilde m})\sum_{i=1}^nm_im_jr_{ij}^2,$ where
$\tilde m$ is the total mass.
So, for fixed $i$ and $j$, we have
$$-{m_im_j\over r_{ij}^2}-b{m_im_j\over r_{ij}^{b+1}}=
2{\sigma\over{\tilde m}}m_im_jr_{ij}.$$
Multiplying by $r_{ij}^{b+1}$, we obtain
$$f(r_{ij}):=2\sigma r_{ij}^{b+2}+{\tilde m}r_{ij}^{b-1}+{\tilde m}b=0.$$
 
Regarding the above equation as a polynomial in the variable
$r_{ij}$, since $\sigma <0$, $f(0)= b{\tilde m}>0$ and the
coefficients polynomial have just one change of sign, we can
verify easily that the function $f$ has exactly one positive
root. Observe that the function $f$ only depends on the
total mass $\tilde m$, and therefore the respective solution for
$f(r_{ij})$ is the same for all mutual distances. We have thus 
proved the following result.
\begin{theo}\label{lagrange}
In the Manev-type three body problem, for any values of the masses, there are exactly two equilateral central configurations, which correspond to the two possible orientations of a triangle in a plane.
\end{theo}

\section{A Framework for the Study of Collisions}
We will further study the dynamics at an near total collision for Manev-type $n$-body problems. A convenient framework for this purpose is given by the so-called McGehee coordinates \cite{mcgehee74},
\begin{equation}\begin{split}\label{mcgehee}
\rho &= ({\bf r}^t M{\bf r})^{1/2} \\
\bf{s} &= \rho^{-1}{\bf r}  \\
v &= \rho^{b/2}({\bf p}^t{\bf s}) \\
\bf{u} &= \rho^{b/2}({\bf p} - ({\bf p }^t{\bf s})M {\bf s}),
\end{split}
\end{equation}
where $M=\mbox{diag}(m_1,m_1,m_2,m_2,\ldots,m_n,m_n)$, ${\bf r}=({\bf r}_1,\ldots,{\bf r}_n)$, and  ${\bf p}=({\bf p}_1,\ldots,{\bf p}_n)$.
After a reparametrization of the time variable,
\begin{equation}\label{repar}
d\tau =r^{-1-b/2}dt,
\end{equation}
the  equations of motion (\ref{hamsys}) become
\begin{equation}\label{eqmot}
\begin{split}
\rho^{\prime } =&\rho v \\
v^{\prime } =&\frac{b}{2}v^{2} + {\bf u}^t M^{-1}{\bf u} -
\rho^{b-1}W({\bf s}) - bV({\bf s})  \\
\bf{s}^{\prime } =& M^{-1}\bf{u}  \\
\bf{u}^{\prime } =& \left(\frac{b}{2}-1\right){\bf{u}}v -
({{\bf u}}^{t}M^{-1}{{\bf u}})M{{\bf s}} +\\
&\rho^{b-1}[W({{\bf s}})M{{\bf s}} + \nabla W({\bf{s}})] +
bV({\bf{s}})M{\bf{s}} + \nabla V({\bf{s}}).
\end{split}
\end{equation}
Here the prime denotes differentiation with respect to the new (fictitious) time variable $\tau $, and the old notation is maintained for the new dependent variables, which are now functions of $\tau $. Furthermore, the new variables fulfill the constraints ${\bf s}^t M{\bf s}=1$ and ${\bf u}^t {\bf s}=0$.

In these coordinates the energy integral (\ref{hamf}) turns into the relation
\begin{equation}\label{eqenerg}
\frac{1}{2}({\bf u}^t M^{-1}{\bf u} + v^{2}) - \rho^{b-1}W({\bf s})
- V({\bf s}) = h\rho^{b}.
\end{equation}
We define the total collision manifold as
\begin{equation}\label{totcoll}
C = \{ (\rho, {\bf s}, v, {\bf u}) | \quad  \rho =0, \quad
{\bf u}^t M^{-1}{\bf u} + v^{2} - 2V({\bf s}) = 0 \}.
\end{equation}

Notice that $\rho^{\prime } =0$ if $\rho=0$, so $C$ (which is 
an analytic submanifold of codimension 1 in the boundary of 
the phase space) is invariant
under the flow of the system (\ref{eqmot}). By continuity of  the solutions respect to initial conditions,
the flow on $C$ provides important information about the orbits close
to triple collision (see \cite{mcgehee74} for more details).
The total collision manifold can also be regarded as an invariant boundary pasted onto each energy surface:
\beq
E_h=\{(\rho,v,{\bf s},{\bf u})|\frac{1}{2}({\bf u}^t M^{-1}{\bf u} + v^{2}) - \rho^{b-1}W({\bf s})
- V({\bf s}) = h\rho^{b}\}.
\eeq
These concepts are ideal for understanding the qualitative behaviour of total- and near-total-collision solutions.
\section{Collision and Near-Collision Dynamics}

In this section we will study the dynamics of total- and near-total-collision orbits of the Manev-type $n$-body problem. An important
role in this study is played by central configurations and by the solutions that can be derived from them.

In the planar Newtonian $n$-body problem, a rigid rotation of a central configuration is called a relative equilibrium; in rotating coordinates, relative equilibria are fixed points. A non-rotating homothetic orbit of a central configuration is called a homothety. The composition of a relative equilibrium and a homothety is called a homographic solution. 

In the Manev-type three-body problem, since the potential only
depends of the bodies' mutual distances, the central configurations 
are invariant under rotations, so any central configuration determines 
a particular periodic orbit, which in a rotating frame is a fixed point. So in the Manev-type three body problem, any central configuration corresponds to a relative equilibrium.

In the Newtonian case, any central configuration also corresponds to a homothetic orbit. But is this valid for Manev-type potentials too? As we will further prove (see Theorem \ref{simu} and Section 8), this property is not satisfied in general. To show this, and to determine under what circumstances homothetic solutions still exist, we will prove several preliminary results. 

Notice that the flow on $C$ is given by the equations
\begin{equation}\label{eqmocol}
\begin{split}
v^{\prime } &=\frac{b}{2}v^{2} + {\bf u}^t M^{-1}{\bf u} -
bV(\bf{s}) \\
{\bf s}^{\prime } &= M^{-1}{\bf u}  \\
{\bf u}^{\prime } &= (\frac{b}{2}-1){\bf u}v -
({\bf u}^{t}M^{-1}{\bf u})M{\bf s} + bV({\bf s})M{\bf s} + \nabla
V({\bf s}).
\end{split}
\end{equation}
The equilibrium points of system (\ref{eqmocol}) are
given by ${\bf u} = 0$, $v=\pm \sqrt{2V({\bf s})}$, where $\bf{s}$
must be a critical point for the function $V({\bf s})$ restricted
to the unit sphere corresponding to the mass matrix $M$.
The masses are involved because the equation $bV({\bf s})M{\bf s} + \nabla V({\bf s})=0$
must be satisfied. But these are the critical points of the function $\tilde V$, which is the restriction of the homogeneous potential $V$ to the unit sphere given by the mass matrix $M$. Such critical points correspond to the central configurations of the homogeneous potential $V$.

\begin{prop}\label{grad-v}
For any value of $b>2$,
the flow on the total collision manifold $C$ is \it{gradient-like}
with respect to the coordinate $-v$ (i.e.\ the flow increases with 
respect to $-v$ along non-equilibrium solutions).
\end{prop}
\begin{proof}
The energy relation (\ref{eqenerg}), restricted to $C$, takes the
form \[{\bf u}^t M^{-1}{\bf u} + v^{2} - 2V({\bf s}) = 0.\] Using the
above expression and (\ref{eqmocol}), we get that 
\begin{equation}\label{gradlike}
v^{\prime } = (1 - \frac{b}{2}){\bf u}^t M^{-1}{\bf u}
\end{equation}
on $C$. If ${\bf u}\neq {\bf 0}$ then $v'<0$ is increasing with respect to $-v$. On the other hand if ${\bf u}={\bf 0}$ then ${\bf u}'=bV({\bf s})M{\bf s} + \nabla
V({\bf s})$, i.e. ${\bf u}'=0$ only if ${\bf s}$ is a critical point of $\tilde V$.
Consequently $-v$ is strictly increasing along nonequilibrium solutions, which means that the vector field is gradient-like with respect to $-v$.
\end{proof}

Denote by $\ind({\bf s}_0)$ the index of the critical point ${\bf s}_0$, i.e.\ the number of eigenvalues of $D^2\tilde V({\bf s}_0)$
with negative real part. Then we can prove the following result.

\begin{theo}
Let ${\bf s}_0$ be a nondegenerate central configuration of the planar $n$-body problem with potential $V$ and $b>2$. Then the dimensions of $W^u({\bf s}^+_0)$ and $W^s({\bf s}^-_0)$ are the same and equal to $2n-2-\ind({\bf s}_0)$ in $E_h$. The dimensions  of $W^s({\bf s}^+_0)$ and $W^u({\bf s}^-_0)$ are the same and equal to $2n-4+\ind({\bf s}_0)$ in $E_h$. The dimension of $E_h$ is $4n-5$.
\end{theo}
\begin{proof}

Let ${\bf s}_0$ be a central configuration, $v =\pm \sqrt{2V({\bf s}_0)}$, and ${\bf u}=0$, then the equation of motion restricted to $E_h$ are
\begin{equation}\label{restricted}
\begin{split}
\rho^{\prime } =&\rho v \\
v^{\prime } =& \left(1-\frac b 2\right ) {\bf u}^t M^{-1}{\bf u}+(b-1)\rho^{b-1}W({\bf s}) +bh\rho^b \\
\bf{s}^{\prime } =& M^{-1}\bf{u}  \\
\bf{u}^{\prime } =& \left(\frac{b}{2}-1\right){\bf{u}}v -
({{\bf u}}^{t}M^{-1}{{\bf u}})M{{\bf s}} +\\
&\rho^{b-1}[W({{\bf s}})M{{\bf s}} + \nabla W({\bf{s}})] +
bV({\bf{s}})M{\bf{s}} + \nabla V({\bf{s}}).
\end{split}
\end{equation}

Taking into account the centre of mass and linear momentum 
integrals as well as the restrictions ${\bf s}^tM{\bf s}=1$ and
${\bf u}^t{\bf s}=0$ of the McGehee coordinates, the above
system has dimension $4n-4$.

Linearizing the system, the eigenvalues for $b>2$ are given by the matrix equation
\beq\label{eigen}
\left (\begin{array}{cccccc}
v&0&0&\ldots&\ldots&0\\
0&0&*&\ldots&\ldots&*\\
\vdots&\vdots& & O_{2n-3}& I_{2n-3} \\
\vdots&\vdots& & A &(\frac b 2-1)v I_{2n-3}\\
0&0 & & & & \\
\end{array}\right)-\mu I_{4n-4}=O_{4n-4},
\eeq
where $I_{N}$ is the $N\times N$ identity matrix, $O_{N}$ is the $N\times N$ zero matrix, A denotes the Hessian matrix of $\tilde V$ (i.e.\ the potential restricted to the sphere of constant moment of inertia) and $*$ denotes an element without importance in the computation of the eigenvalues.

It is clear that the first two eigenvalues are $v\neq0$ (since $V({\bf s}_0)\neq 0)$ and $0$.
To obtain the remaining eigenvalues  of equation (\ref{eigen}), suppose ${\bf z}$ is a $(2n-3)$-vector satisfying
\beq
A{\bf z}=\lambda_i {\bf z}
\eeq
for $i=1,\ldots,2n-3$, $i$ fixed, where $\lambda_1,\ldots,\lambda_{2n-3}$ are the eigenvalues of $A$. Then
\[
\left(
\begin{array}{cc}
O_{2n-3}&I_{2n-3}\\
A& (b/2-1)vI_{2n-3}
\end{array}\right)
\left (\begin{array}{c}
{\bf z}\\
\mu {\bf z}
\end{array}\right)=
\left(
\begin{array}{c}
\mu {\bf z}\\
\{\lambda_i+(b/2-1)v\mu\}{\bf z}
\end{array}
\right).
\]
Consequently $\mu$ is a root of equation (\ref{eigen}) if
\[
\mu^2-(b/2-1)v\mu-\lambda_i=0,
\]
which gives
\[\mu_i^{1,2}=\frac 1 4 \{(b-2)v\pm \sqrt{(2-b)^2v ^2+16\lambda_i}\}\]
for $i=1,\ldots,2n-3$.

Then if $v=\sqrt{V({\bf s}_0)}$, the differential matrix of the vector field restricted to $E_h$ has $2n-2-\ind({\bf s}_0)$ eigenvalues with positive real part and $2n-4+\ind({\bf s}_0)$ with negative real part. The values of the dimensions are switched if $v=-\sqrt{V({\bf s}_0)}$.
\end{proof}

\begin{theo}
Let ${\bf s}_0$ be a central configuration of the collinear $n$-body problem with potential $V$ and $b>2$. Then the dimensions of $W^u({\bf s}^+_0)$ and $W^s({\bf s}^-_0)$ are the same and equal to $n-1$ in $E_h$. The dimensions  of $W^s({\bf s}^+_0)$ and $W^u({\bf s}^-_0)$ are the same and equal to $n-2$ in $E_h$. The dimension of $E_h$ is $2n-3$.
\end{theo}

\begin{proof}
The proof is similar to the one of the previous theorem. Let ${\bf s}_0$ be a central configuration, $v=\pm\sqrt{2V({\bf s}_0)}$. The equations of motion restricted to $E_h$ are given by equation (\ref{restricted}), with the obvious modifications.

Linearizing the system, the eigenvalues in the case $b>2$ are given by the following matrix equation
\beq
\left(\begin{array}{cccccc}
v&0&0&\ldots&\ldots&0\\
0&0&*&\ldots&\ldots&*\\
\vdots&\vdots& & O_{n-2}& I_{n-2} \\
\vdots&\vdots& & A &(\frac b 2-1)v I_{n-2}\\
0&0 & & & & \\
\end{array}\right)-\mu I_{2n-2}=O_{2n-2},
\eeq
where $I_{N}$ and $O_{N}$ are defined as before. Again $A$ is the Hessian matrix of $\tilde V$ and $*$ denotes an element without importance in the computation of the eigenvalues.
The first two eigenvalues are $v\neq0$ (since $V({\bf s}_0)\neq 0)$ and $0$. If $\lambda_1,\ldots,\lambda_{n-2}$ be the eigenvalues of $A$,  then
\[\mu_i^{1,2}=\frac 1 4 \{(b-2)v\pm \sqrt{(2-b)^2v ^2+16\lambda_i}\}\]
for $i=1,\ldots,n-2$.
Note that, in this case, Lemma 4 implies that $A$ is positive definite, and thus the eigenvalues $\lambda_1,\ldots,\lambda_{n-2}$ are all positive. Consequently, for $v>0$, $\mu_i^1$ is negative and $\mu_i^2$ positive. However, for $v<0$, $\mu_i^1$ is positive, whereas $\mu_i^2$ is negative. This concludes the proof.
\end{proof}

We will further state and prove a result that clarifies under what
circumstances homothetic solutions exist.

\begin{theo}\label{simu}
A solution of the Manev-type $n$-body problem is homothetic if and only if the particles form, at  all times, a simultaneous central configuration for the potentials $V$ and $W$.
\end{theo}
\begin{proof}
Assume that the solution is homothetic, then ${\bf s}\equiv{\bf s}_0$, where ${\bf s}_0$ is a constant. Therefore ${\bf s}'\equiv 0$ and, from the second of equations (\ref{eqmocol}), ${\bf u}\equiv 0$.
Thus the homothetic orbits are confined to the invariant plane
\beq
{\mathcal P}=\{(\rho,{\bf s},v,{\bf u})|{\bf s}={\bf s}_0, {\bf u}=\bf {0}\}.
\eeq
So ${\bf u}'=0$ implies that $\rho^{b-1}[W({\bf s}_0)M{\bf s}_0+\nabla W({\bf s}_0)]+bV({\bf s}_0)M{\bf s}_0+\nabla V({\bf s}_0)=0$.
If $\rho$ is not constant then there are $\rho_1\neq 0$ and $\rho_2\neq 0$ with $\rho_1\neq \rho_2$ such that
\beq
\begin{split}
 &\rho_1^{b-1}[W({\bf s}_0)M{\bf s}_0+\nabla W({\bf s}_0)]=-[bV({\bf s}_0)M{\bf s}_0+\nabla V({\bf s}_0)]\\
 &\rho_2^{b-1}[W({\bf s}_0)M{\bf s}_0+\nabla W({\bf s}_0)]=-[bV({\bf s}_0)M{\bf s}_0+\nabla V({\bf s}_0)].
\end{split}
\eeq
This means that $[bV({\bf s}_0)M{\bf s}_0+\nabla V({\bf s}_0)]=0$ and $[W({\bf s}_0)M{\bf s}_0+\nabla W({\bf s}_0)]=0$, i.e. that ${\bf s}_0$ is a simultaneous central confiiguration for the potentials $V$ and $W$.
If $\rho$ is constant, $\rho'=0$ and either $\rho\equiv 0$ or $v\equiv 0$. The first case is trivial, whereas in  the latter case $-\rho^{b-1}W({\bf s})-bV({\bf s})=0$. But this is impossible since $V>0$ and $W>0$.

If ${\bf s}\equiv {\bf s}_0$ is, at all times, a simultaneous central configuration for $V$ and $W$, then the solution is obviously homothetic.
\end{proof}

The next result proves the existence and uniqueness of heteroclinic homothetic solutions for $b>1$.

\begin{theo}
Let ${\bf s}_0$ be a simultaneous  central configuration for the potentials $V$ and $W$. Then, if $b>1$, every energy surface of negative constant ($h<0$), contains a unique homothetic solution defined on $(-\infty,\infty)$, satisfying ${\bf s}={\bf s}_0$ for all times and such that $\rho(\tau)\rightarrow 0$ when $\tau\rightarrow\pm \infty$.
In other words the solution begins and ends in a total collapse, maintaining for all times the same central configuration.\label{hetero}
\end{theo}
\begin{proof}
 Since  ${\bf s}_0$ is a simultaneous central configuration for the potentials $V$ and $W$, we have that $[bV({\bf s}_0)M{\bf s}_0+\nabla V({\bf s}_0)]=0$ and $[W({\bf s}_0)M{\bf s}_0+\nabla W({\bf s}_0)]=0$.
Consequently, the set
\[
{\mathcal P}=\{(\rho,{\bf s},v,{\bf u})|{\bf s}={\bf s}_0, {\bf u}=\bf {0}\}
\]
is invariant for the equations of motion. Restricting these equations to ${\mathcal P}$, we get
\beq\begin{split}\label{add}
&\rho'=\rho v\\
&v'=\frac b 2 v^2-\rho^{b-1}W({\bf s}_0)-bV({\bf s}_0),
\end{split}\eeq
while the energy relation becomes
\[
\frac 1 2 v^2-\rho^{b-1}W({\bf s}_0)-V({\bf s}_0)=h\rho^b.
\]
Equations (\ref{add}) become
\[\begin{split}
&\rho'=\rho v\\
&v'=(b-1)\rho^{b-1}W({\bf s}_0)+b\rho^bh.
\end{split}\]
This leads to
\[
\frac{dv}{d\rho}=\frac{1}{v}[(b-1)\rho^{b-2}W({\bf s}_0)+b\rho^{b-1}h],
\]
which yields
\beq\frac{v^2}{2}=\rho^{b-1}W({\bf s}_0)+\rho^b h+K,
\eeq
where, if $b>1$, we choose $K=V({\bf s}_0)$. If $b>1$ and $h\geq 0$, then $|v|\geq \pm \sqrt{2V({\bf s}_0)}$ and the homothetic orbits are not heteroclinic. If $h<0$, there is a unique curve connecting the points $(\sqrt{2V({\bf s}_0)},0)$ and $(-\sqrt{2V({\bf s}_0)},0)$ on the plane $(v,\rho)$. These facts prove the theorem.
\end{proof}

The following result shows that the above property is also true for
the equilateral central configurations.

\begin{cor}
Let ${\bf s}_0$ be an equilateral central configuration for the potential $U$. Then, if $b>1$, every energy surface of negative constant ($h<0$) contains a unique heteroclinic homothetic solution.
\end{cor}
\begin{proof}
Clearly ${\bf s}_0$ is a simultaneous central configuration for the potentials $V$ and $W$. The proof follows from Theorem \ref{hetero}.
\end{proof}

Two submanifolds ${\mathcal E}_1$ and $\mathcal{E}_2$ of a submanifold ${\mathcal E}$ are said to be transverse at a point $x$ if one of the following situation arises:
\begin{enumerate}
\item ${\mathcal E}_1 \cap{\mathcal E}_2=\emptyset$;
\item $x\in{\mathcal E}_1 \cap{\mathcal E}_2$ and $T_x{\mathcal E}_1+T_x{\mathcal E}_2=T_x{\mathcal E}$, where $T_x{\mathcal E}$ denotes the tangent space to ${\mathcal E}$ at the point $x$.
\end{enumerate}
We can now prove the following result:

\begin{theo}
In the  planar Manev-type $n$-body problem with $b>2$, a necessary condition for having a transversal homothetic
solution $\gamma_h({\bf s}_0)$ in $E_h$ with $h<0$ is that $\tilde
V$ be a non-degenerate minimum at the point ${\bf s}_0$ associated
with the homothetic solution.
\end{theo}
\begin{proof}
Let $\gamma_h({\bf s}_0)$ be a transversal homothetic solution in
$E_h$ with $h<0$. Then by Theorem 3 both $W^u({\bf s^+}_0)$
and $W^s({\bf s^-}_0)$ are $(2n-2-\ind({\bf s}_0))$-dimensional
and $E_h$ is $(4n-5)$-dimensional. Since $\gamma_h({\bf s}_0)\in
W^u({\bf s}_0)\cap W^s({\bf s}_0)$ and $\gamma_h({\bf s}_0)$ is
transversal we have that
\[
\dim E_h\leq \dim W^u({\bf s}_0)+\dim W^s({\bf s}_0)-1.
\]
That is $4n-5 \leq 4n-5-2~\ind({\bf s}_0)$. Therefore $\ind({\bf s}_0)=0$
and the function $\tilde V$ has a nondegenerate minimum at ${\bf s}_0$.
\end{proof}

\section{Simultaneous Configurations and Relative Equilibria}
In this closing section, we will show that, for most choices of the masses in the quasihomogeneous 3-body problem, the collinear central configurations of the potential $U$ are not simultaneous relative equilibria for $V$ and $W$.

\begin{theo}\label{nowheredense}
 Let $\Sigma_3$ be the set of masses $(m_1,m_2,m_3)\in{\mathbb R}_+^3$ for which the collinear configurations are simultaneous central configurations for the potentials $V$ and $W$. Then the set $\Sigma_3$ is nonempty and nowhere dense in ${\mathbb R}_+^3$.
\end{theo}
\begin{proof}
Assume the configuration ${\bf s}_V(m_1,m_2,m_3)$ is a collinear central configuration for $V$ and ${\bf s}_W(m_1,m_2,m_3)$  is a collinear central configuration for $W$. In \cite{euler1767}, Euler found
a complicated formula that expresses the ratio of the distances
between the masses for any rectilinear central configuration
in the Newtonian case. Euler's formula can be directly extended 
to any homogeneous potential. Moreover, the fact that Euler's
expression is an analytic function of the masses remains true
in the homogeneous case. Therefore both ${\bf s}_V$ and
${\bf s}_W$ are analytic functions of $m_1, m_2$ and $m_3$,
as long as the masses are positive.
Consequently the function  $z={\bf s}_V-{\bf s}_W$ is also
an analytic function of the masses.

For the function $V$, Euler's formula depends on $a$, whereas for $W$ it depends on $b$. So in general ${\bf s}_V\neq {\bf s}_W$, therefore for every $a$ and $b$ with $a\neq b$ there are values of the masses for which $z\neq 0$. Since $z$ is a nonzero analytic function, its zeroes form a nowhere dense set. 

The nonemptiness of the set of simultaneous central configurations follows from noticing that if $m_1=m_3$ and the mass $m_2$ is located halfway between the other two, then the three masses form a simultaneous central configuration for $V$ and $W$. 

\end{proof}

A consequence of Theorems \ref{simu} and \ref{nowheredense} is that, for most values of the masses, there are no rectilinear homothetic orbits. More precisely:

\begin{cor}
If $(m_1,m_2,m_3)\in {\mathbb R}_+^3\setminus \Sigma_3$, then there are no rectilinear homothetic orbits.
\end{cor}

This shows that the rectilinear homothetic orbits are characteristic
to homogeneous potentials, but they prove unlikely in the
quasihomogeneous case.



\begin{thebibliography}{2004}
{
\bibitem{abraham}
R.~Abraham and J.E.~Marsden, {\it Foundation of Mechanics} (2nd ed.), Benjamin, New York, 1978.\\

\bibitem{CLP03} Corbera, M., Llibre, J. and P\'erez--Chavela, E.,
{\it Equilibrium points and central configurations for the
Lennard-Jones $2$- and $3$-body problems}, {\it Celestial 
Mechanics and Dynamical Astronomy} {\bf 89}(3), 235-266, (2004).\\

\bibitem{craig99}
S.~Craig, F.~Diacu, E.A.~Lacomba and E.~P\'erez-Chavela, {\it On
the anisotropic Manev problem}, Journ.\ Math.\ Physics
{\bf 40}, 1-17, (1999).\\

\bibitem{diacu96}
F.~Diacu, {\it Near-collision dynamics for particle systems with quasihomogeneous potentials}, Journ.\ Differential Equations {\bf 128}, 58-77, (1996).\\

\bibitem{diacu00} F.~Diacu, V.\ Mioc, and C.\ Stoica, {\it Phase-space structure and regularization of Manev-type problems}, Nonlinear Analysis {\bf 41}, 1029-1055, (2000).\\

\bibitem{diacu05}
F.~Diacu, E.~P\'erez-Chavela and M.~Santoprete, {\it Saari's
conjecture in the collinear case}, Transactions AMS (to appear 
in October 2005).\\

\bibitem{euler1767}
L.~Euler, {\it De motu rectilineo trium corporum se mutuo
attrahentium}, Novi Commentarii Academiae Scientarum
Petropolitanae {\bf 11}, 144-151, (1767).\\

\bibitem{mcgehee74}
R.~McGehee, {\em Triple collision in the collinear three-body
problem}, Inventiones Math., \textbf{27}, 191-227, (1974).\\

\bibitem{meyer92}
K.~Meyer and G.~Hall, \emph{Introduction to Hamiltonian Dynamical
Systems and the $\,n$-Body Problem}, Applied Mathematical Science \textbf{90}, Springer-Verlag 1992.\\

\bibitem{moulton10}
F.R.~Moulton, \emph{The straight line solutions of the problem
of $N$ bodies}, Annals of Math., \textbf{12}, 1-17, (1910).\\

\bibitem{ernesto96}
E.~P\'erez-Chavela, D.~Saari, A.~Susin and Z.~Yan, \emph{Central
configurations in the charged three body problem},
Contemporary Mathematics \textbf{198}, 137--155, (1996).\\

\bibitem{ernesto98}
E.~P\'erez-Chavela and L.~Vela-Ar\'evalo, \emph{Triple collision
in the quasihomogeneous collinear three-body problem}, Journ.\
Differential Equations \textbf {148}, 186-211, (1998).\\

\bibitem{saari80}
D.~Saari, \emph{On the role and properties of $\,n$-body central
configurations}, Celestial Mechanics \textbf{21}, 9--20, (1980).\\

\bibitem{smale71}
S.~Smale, \emph{Problems on the nature of relative equilibria in
celestial mechanics}, Lecture Notes in Math., \textbf{197}, 194--198, (1971).\\

\bibitem{smale98}
S.~Smale, \emph{Mathematical Problems for the Next Century},
Mathematical Intelligencer, \textbf{20}, No. 2, 7--15, (1998).\\

\bibitem{wintner41}
A.~Wintner, \emph{The Analytical Foundations of Celestial
Mechanics}, Princeton Univ.\ Press, Princeton,
N.J., 1941.\\}

\end{thebibliography}
\end{document}